\documentclass[acmsmall,nonacm]{acmart}
\AtBeginDocument{%
  }

\usepackage{enumitem}
\usepackage{algorithm}
\usepackage{algorithmic}
\usepackage{mathtools}
\usepackage{subcaption}

\usepackage{listings}
\usepackage{xcolor}

\lstdefinelanguage{ASP}{
  morekeywords={not},
  sensitive=true,
  morecomment=[l]\%,
  morestring=[b]"
}

\newcommand{\B}[0]{\mathfrak{B}}

\newcommand{\fbar}[1]{f_{#1}^u}
\newcommand{\SPbar}[1]{(S_P)_{#1}^u}
\newcommand{\SPubar}[1]{(S_P)_{#1}^{\ell}}

\newcommand{\xbar}[1]{x_{#1}^u}
\newcommand{\fubar}[1]{f_{#1}^{\ell}}

\newcommand{\xubar}[1]{x_{#1}^{\ell}}
\newcommand{\lfp}[1]{\operatorname{lfp}(#1)}
\newcommand{\lfpB}[1]{\operatorname{lfp}_B(#1)}
\newcommand{\lfpS}[2]{\operatorname{lfp}_{#2}(#1)}

\newcommand{\gfp}[1]{\operatorname{gfp}(#1)}
\newcommand{\gfpB}[1]{\operatorname{gfp}_B(#1)}
\newcommand{\gfpS}[2]{\operatorname{gfp}_{#2}(#1)}
\newcommand{\bbound}[0]{[x^{\ell},x^u]}
\newcommand{\ok}[1]{\mathbf{ok}_{#1}}
\newcommand{\fix}[1]{\operatorname{Fix}(#1)}

\begin{document}

\title{Bounding Fixed Points of Non-Monotone Processes: Theory to Practice}

\author{Abdullah Rasheed}
\email{arasheed@utexas.edu}
\orcid{0009-0004-7475-9260}
\affiliation{%
  \institution{The University of Texas at Austin}
  \city{Austin}
  \state{Texas}
  \country{USA}
}

\author{Vijay K. Garg}
\email{garg@ece.utexas.edu}
\orcid{0000−0002−5797−4389}
\affiliation{%
  \institution{The University of Texas at Austin}
  \city{Austin}
  \state{Texas}
  \country{USA}
}


\begin{abstract}
  Many modern solvers and program analyzers rely on non-monotone reasoning (e.g. negation-as-failure, speculative updates, backtracking) for which classical monotone fixed-point methods do not apply. The general problem of finding the fixed points of these processes is a difficult one. For this reason, there have been theoretical efforts in existing Approximation Fixpoint Theory (AFT) from the domain of logic programming to approximate fixed points of non-monotone operators. Tight approximations of these fixed points are highly useful for accelerating non-monotonic computations by restricting the search space. In practice, however, even the best approximations obtained through AFT can be coarse and computationally expensive. We aim to address both issues to make AFT approximation methods practical for use in programming languages (PL) settings. To mitigate inefficiency, we prove the soundness of an abstract interpretation for approximating operators. To improve upon coarse approximations, we carefully introduce controlled unsoundness to design an effective yet practical algorithm for partitioning and tightening AFT’s best approximations. This algorithm is sound, anytime, and guarantees termination on finite-height lattices. We further present a modification that ensures polynomial-time complexity. We instantiate these methods in two settings: (1) answer set programming, where it serves as a convergence-accelerating pre-processor, and (2) speculative program analysis, where it reduces rollback while preserving soundness. In both settings, we focus on implementation-level details to demonstrate the practical applicability of our methods.
\end{abstract}




\maketitle

\section{Introduction}

Many analyses and algorithms in programming languages are designed around \emph{monotonic} computation. Monotonicity provides a well-understood foundation: once a fact is derived, it remains true, and states evolve along a chain $S_1 \sqsubseteq S_2 \sqsubseteq \cdots$ that reaches a fixed point. This property underlies a vast range of techniques from dataflow analysis and abstract interpretation to type inference and constraint solving. However, monotonicity also comes at a cost: once a decision is made, it cannot be retracted. This rigidity limits the ability of such systems to model \emph{speculative reasoning} or \emph{reversible computation}, both of which are increasingly common in modern algorithms and implementations.

In contrast, \emph{non-monotonic} processes allow information to be withdrawn or revised. Such behavior arises naturally in domains that rely on backtracking, revision, or speculation. For example, the decision process of Conflict-Driven Clause Learning (CDCL) SAT solvers backtracks when a decision leads to a conflict \cite{marques2009}; speculative optimizations in dynamic compilers and program analyses assume properties that may later be invalidated \cite{holzle1992,fluckiger2017,devecsery2018}; and non-monotonic reasoning forms the basis of Answer Set Programming (ASP) which is widely utilized in knowledge representation and has been used in program analysis, model checking, constraint solving, and much more \cite{marek1997,erdem2016,heljanko2003,dawson1996}. These examples demonstrate that, in practical applications, non-monotonicity is not a defect to be avoided but a powerful mechanism for adaptability and efficiency.

Speculative program analyses are modeled in this paper as a concrete instance of non-monotonic behavior. Such analyses often assume facts that may later be invalidated as new information becomes available, forcing them to retract or revise earlier speculations. We demonstrate that there exists a non-monotone (in fact, anti-monotone) operator whose fixed points coincide exactly with the sets of ``stable'' assumptions. That is, an operator $\Phi$ that makes necessary revisions to the set of assumptions $\sigma$ and whose fixed points $\sigma = \Phi(\sigma)$ are the sets that are not refuted in the speculative analysis computed by assuming $\sigma$. Furthermore, with a larger set of assumptions, a more precise analysis is computed, which in turn can disprove more assumptions. It is then evident that $\Phi$ is non-monotone (in fact, anti-monotone). Identifying these stable sets of assumptions enables analyses to improve precision and optimize performance. For instance, a dynamic analysis can apply a given set of stable assumptions to reduce unnecessary retractions or deoptimizations caused by runtime counter-examples.

Despite its practical relevance, the approximation of fixed points of non-monotone functions has received little attention in programming languages. Existing approaches rely heavily on monotonicity to guarantee convergence and soundness, while non-monotonicity is often feared for its lack thereof. 
This paper bridges that gap by taking advantage of existing theory and introducing new results to present practical and effective approaches for approximating the fixed points of non-monotone functions in PL settings while improving upon theoretical barriers.

\paragraph{Contributions} Our main contributions are as follows:
\begin{itemize}
  \item \textbf{AFT with Abstractions.} We prove a new soundness theorem for abstract interpretation of non-monotone fixed point approximation, enabling quicker computation of AFT operators abstractly.
  \item \textbf{Branch-and-Bound Approximations.} We devise an effective \emph{anytime} branch-and-bound style algorithm that often produces a set of tighter approximations than AFT prescribes. In particular, it finds a set of approximations that spans all fixed points such that the total size of all bounds is bounded from above by the size of the best possible approximation guaranteed by AFT. We establish termination (for finite height lattices) and soundness guarantees while retaining practical efficiency by leveraging structural properties of a generalized AFT operator to prune the search space. We also introduce modifications that guarantee polynomial time complexity in the height of the lattice (when the lattice is finite height and the approximation refinement is efficient, enabled by lattice abstractions, structural properties of the given function, or a restricted refinement).
  \item \textbf{Feasibility and Demonstration in PL Practice.} We model speculative program analysis in this paradigm, showing that the approach integrates cleanly in PL settings and the algorithms deliver tangible benefits, thereby evidencing practical feasibility for such tasks. We focus on implementation-level details, outlining the trade-offs with integration, and the use cases where these methods are most effective.
  \item \textbf{Use Case for Pre-Processing.} We demonstrate, using the example of ASP solving, how our methods can be used as a pre-processing step for problem domains in our paradigm by substantially reducing the search space size, thereby improving efficiency. We again focus on implementation-level details and illustrate how our proposed algorithm can reduce the workload of existing ASP solvers.
\end{itemize}

\section{Related Work}
\label{sec:related-works}

The study of approximating fixed points of non-monotone operators is not a new one. In fact, an Approximation Fixpoint Theory (AFT) was proposed long ago for the purposes of analyzing, defining, and approximating semantics of non-monotone logics \cite{denecker2000}. Since its introduction, the theory has continued to develop and find use in the areas of knowledge representation and logic programming \cite{denecker2012,strass2013,liu2022}. AFT is an algebraic framework for the approximation of fixed points using ``approximation operators'' over a bilattice of pairs, with the property that they bound the fixed points of the original, possibly non-monotone function in an interval. In \cite{denecker2004}, the notion of an ``ultimate approximation'' is introduced as a necessarily existing unique and most precise (partial) approximation. Despite being the best approximation, it may still a coarse one (possibly even spanning the entire lattice). In this paper, we restate the ultimate approximation in a more concrete and compact manner, demonstrate its soundness for approximating subsets of fixed points in subspaces, and crucially, take advantage of these properties to design an algorithm that improves the bound beyond the ultimate approximation. 

AFT originated in the world of knowledge representation, and has been applied to the stable model and well-founded semantics of logic programs \cite{denecker2004,denecker2000,gelfond1988,gelder1991}. We further extend this example by applying our branch-and-bound style heuristic to narrow in on even tighter approximations for stable models in logic programs. Furthermore, we take a closer look at the implementation level details to see how our methods can serve as a pre-processor for speeding up ASP solvers. The work done in \cite{eiter2004} particularly enables the integration of these methods by providing equivalence guarantees under specific ``partial evaluations'' of a logic program.

To bridge this theory from the domain of knowledge representation to the practice of programming languages, we show that program analysis is a suitable use-case. Furthermore, we revisit the \emph{optimistic hybrid analysis} proposed in \cite{devecsery2018}. Their approach infers a set of ``likely invariants'' from a sample of executions. These likely invariants are then fed as assumptions to a static analyzer to produce an \emph{unsound} yet more precise analysis whose assumptions are later verified dynamically. However, as outlined in their paper, when those assumptions fail, the runtime system bears the cost. Optimizations are rolled back and the program must re-execute. We formalize the problem of finding a consistent assumption set by showing that it reduces to finding fixed points of a non-monotone operator. Depending on the application, these stable sets can represent assumptions that are not provably inconsistent or those provably consistent under speculative analysis. Our heuristic can be used to help effectively identify such stable sets, reducing dynamic rollback and improving predictability.

\section{Background}

A partially ordered set $(L, \leq)$ is a \emph{complete lattice} if every subset $S \subseteq L$ has both a least upper bound (supremum, or join, denoted $\bigsqcup S$) and a greatest lower bound (infimum, or meet, denoted $\bigsqcap S$). Every complete lattice has a top element ($\top = \bigsqcup L$) and a bottom element ($\bot = \bigsqcap L$). We also use $\sup$ and $\inf$ to represent $\bigsqcup$ and $\bigsqcap$ respectively. Two elements $x,y \in L$ are \emph{incomparable} if $x \nleq y$ and $y \nleq x$. This is denoted by $x||y$. The height of a lattice is $\ell(L) = \sup\{k \mid \bot_L = x_0 \leq \cdots \leq x_k = \top_L\}$. If $\ell(L)$ is not defined, then we may write $\ell(L) = \infty$.

Given two posets $(P, \leq_P)$ and $(Q, \leq_Q)$, a function $f: P \to Q$ is \emph{monotone} (or order-preserving) if for all $x, y \in P$, $x \leq_P y$ implies $f(x) \leq_Q f(y)$. On the other hand, $f$ is \emph{anti-monotone} if for all $x,y \in P$, $x \leq_P y$ implies $f(x) \geq_Q f(y)$.

\begin{theorem}[Knaster-Tarski Fixed-Point Theorem]
\label{thm:knaster-tarski}
Let $(L, \leq)$ be a complete lattice and $f: L \to L$ be a monotone function. The set of fixed points of $f$ in $L$ is a non-empty complete lattice. It contains a least fixed point $\lfp{f}$ and a greatest fixed point $\gfp{f}$, given by:
$$ \lfp{f} = \bigsqcap \{x \in L \mid f(x) \leq x\} $$
$$ \gfp{f} = \bigsqcup \{x \in L \mid x \leq f(x)\} $$
\end{theorem}

An element $x \in L$ is called a \emph{pre-fixed point} if $f(x) \leq x$, and it is called a \emph{post-fixed} point if $x \leq f(x)$. Then, $\lfp{f}$ is the infimum of the pre-fixed points and $\gfp{f}$ is the supremum of the post-fixed point when $f$ is monotone. We will denote the set of fixed points of $f$ by $\fix{f} = \{x \in L \mid x = f(x)\}$.

A \emph{chain} is a subset $S \subseteq L$ such that for all $x,y \in S$, $x \leq y$ or $y \leq x$. A chain $\{x_k\}$ is ascending if $x_1 \leq x_2 \leq \ldots$, and it is descending if $x_1 \geq x_2 \geq \ldots$. 

An ascending or descending chain $\{x_k\}$ stabilizes if there is some $j$ such that $x_j = x_{j+1} = \ldots$. A lattice $L$ satisfies the \emph{Ascending Chain Condition} (ACC) if every ascending chain stabilizes, and it satisfies the \emph{Descending Chain Condition} if every descending chain stabilizes.

Let $C$ and $A$ be complete lattices. Then, if $\alpha : C \to A$ and $\gamma : A \to C$ are monotone, $\alpha \dashv \gamma$ iff $\alpha(x) \leq_A y \Leftrightarrow x \leq_C \gamma(y)$ ($\alpha$ is called the \emph{left-adjoint} and $\gamma$ is called the \emph{right-adjoint}). This pair of functions forms a \emph{Galois connection}. It is necessarily true that $\alpha$ preserves arbitrary joins and $\gamma$ preserves arbitrary meets. If $\alpha\gamma = \mathrm{id}_A$ (equivalently, $\alpha$ is surjective, and also equivalently, $\gamma$ is injective), then they form a \emph{Galois insertion}. 

\subsection{Answer Set Programming}

A (logic) program $P$ is a set of rules of the form $$A \gets L_1,\ldots,L_m$$ where $A$ is an atom and each $L_i$ is a literal (either a positive or a negative atom). The rule states, in an informal sense, that if $L_1,\ldots,L_m$ hold, then $A$ should also hold in a given interpretation of the ground atoms. A program $P$ is \emph{non-stratified} if it has recursion through negation (the exact definition will not be important in this paper). These programs are provably harder than stratified ones, and their main reasoning tasks are NP/coNP-complete \cite{eiter1995,marek1999}. The Herbrand base $H_P$ of $P$ is the set of all ground atoms occuring in $P$ (after grounding all variables), and an Herbrand model $M$ of $P$ is a subset of the Herbrand base such that setting the atoms of $M$ to true (and everything else to false; two-valued logic) satisfies the rules of $P$. For example, in the program
\begin{align*}
    a &\gets \\
    b &\gets a,c
\end{align*}
the Herbrand base is $\{a,b,c\}$ and the Herbrand models are $\{a\},\{a,b\},\{a,b,c\}$. Observe that $\{\}$ is \emph{not} an Herbrand model because the rule $a \gets$ requires every Herbrand model to contain $a$. Also, $\{a,c\}$ is \emph{not} an Herbrand model because the rule $b \gets a,c$ states that, since we have $a,c$, we should also have $b$.

An Herbrand model $M$ is \emph{minimal} if no strict subset of $M$ is an Herbrand model.

Given any subset $M$ of the Herbrand base, the Gelfond-Lifschitz reduction $P_M$ of the program $P$ is as follows:
\begin{enumerate}
    \item for all $A \in M$, remove all rules that contain $A$ negatively in the body (i.e. any rule of the form $B \gets \ldots ,\neg A,\ldots$).
    \item remove all negative literals in the bodies of all rules of $P$
\end{enumerate}
The reduct $P_M$ is always a \emph{positive} program. That is, it contains no negative literals and thus has a unique minimal Herbrand model. The unique minimal Herbrand model of a positive program can be found by repeatedly applying the rules of the program to the empty set until a fixed point is reached. The unique minimal Herbrand model of $P_M$ is denoted by $S_P(M)$. We call $S_P$ the Gelfond-Lifschitz (GL) operator (parameterized by $P$). Note that $S_P$ is anti-monotone since for larger $M$, the reduct $P_M$ is smaller.

A subset $M$ of the Herbrand base is called a \emph{stable model} (or \emph{answer set}) of $P$ if the unique minimal Herbrand model of $P_M$ is equal to $M$ (it is a fixed point of $S_P$). Stable models are ``self-supporting,'' and are of great importance in answer set programming. For example, when modeling combinatorial optimization problems, stable models often correspond to minimal solutions. This is shown in the following theorem.

\begin{theorem}
    Any stable model of $P$ is a minimal Herbrand model of $P$ \cite{gelfond1988}.
\end{theorem}

\section{Initial Monotone Approximations}
\label{sec:init-approx}

To approximate the fixed points of a (possibly) non-monotone operator, we approximate it from above and below. In particular, we approximate it using monotone envelopes for ease of computation. This is effectively a more concrete and digestible restatement of the ultimate approximation \cite{denecker2004}. The discussion in this section is valuable, but generalization begins at Section \ref{sec:soundness-thms} and results needed for our later algorithm subsequently follow.

For a given function, we are interested in the \emph{least} monotone over-approximation. This will allow us to use a monotone approximating function to find an initial approximation that will be refined. Let $L$ be a complete lattice and $f : L \to L$. We then define \[f^u(x) = \sup_{y \leq x} f(y)\] and observe that $f^u$ is the pointwise least (the \emph{best}) monotone function over-approximating $f$ \cite{cousot1979}. Our eventual goal will be to iteratively restrict the domain for which $f^u$ is monotone to improve the approximation.

To do this, we also utilize a best monotone under-approximation of $f$. Define $f^{\ell}$ as
\[
f^{\ell} (x) = \inf_{y \geq x} f(y)
\]
and observe the dual property that $f^{\ell}$ is the greatest (the \emph{best}) monotone under-approximation of $f$. These envelopes $f^{\ell},f^u$ correspond to the lower/upper borders of AFT's ultimate approximation. To approximate solutions to $x = f(x)$, these approximations will work dually together to improve each other's approximation by \emph{restricting} the domain over which the other function needs to be monotone. For the sake of simplicity throughout this paper, we will assume that $f$ indeed has a fixed point.

We now have best upper and lower bound monotone functions $f^u, f^{\ell}$ for $f$. For the purposes of fixed point computation/approximation, there are a variety of methods known for monotone functions. We will use the easier problem of finding the least fixed point and greatest fixed point of monotone functions (namely, $f^u$ and $f^{\ell}$) to approximate the fixed points of the non-monotone function $f$.
As we will soon see, this approximation will be a bound containing all fixed points of $f$ (given they exist).  This final bound provides a restricted search space for the fixed points of $f$ along with a verifiable certificate of ``closeness'' to a fixed point.

Also observe that, at face value, computing $f^u$ and $f^{\ell}$ is computationally infeasible. It may require taking the supremum/infimum over a very large subset, which can be just as expensive to compute as a brute force search. In the examples presented throughout this paper, we observe that in many scenarios (when the envelopes have closed-form solutions), these approximating monotone functions are much more computationally feasible than they initially seem. Anti-monotone functions will later be of particular interest. To remedy the cases that the approximating functions are \emph{not} computationally feasible, we later prove that an abstract interpretation can be induced for the approximation operator while maintaining soundness in the resulting bounds.

\begin{example}
    \label{example:step-1}
    Throughout this paper, we will demonstrate the refinement process on a non-monotone function $f$ that represents the ``best response'' of two actors in a system. The fixed points of this function can be thought of as the Nash Equilibrium of a game with two players.

    Let $L = [0,1] \times [0,1]$ ordered component-wise, and let $$f(x_1,x_2) = (BR_1(x_2),BR_2(x_1))$$ be a function where $x_1,x_2$ are the current fractions of resources that player 1 and player 2 have claimed respectively. The ``best response'' of player 1 is $BR_1(x_2) = 1-x_2^2$ (if player 2 takes more resources, player 1 takes quadratically less). The ``best response'' of player 2 is $BR_2(x_1) = x_1/2$ (if player 1 claims more resources, player 2 claims a fraction of them).

    Observe that $f$ is non-monotone: $(0.5,0.5) \leq (1,1)$ but $f(0.5,0.5) = (0.75,0.25)$ and $f(1,1) = (0,0.5)$. To simplify the approximating monotone functions, first observe that $f(x_1,x_2)$ is maximized when $x_2$ is minimized and $x_1$ is maximized. Then, the over-approximating monotone function simplifies to
    \[
    f^u(x_1,x_2) = \sup_{(y_1,y_2) \leq (x_1,x_2)} f(y_1,y_2) = f(x_1,0)
    \]

    For the under-approximating function, observe that $f(x_1,x_2)$ is minimized when $x_2$ is maximized and $x_1$ is minimized. Then the under-approximating monotone function simplifies to
    \[
    f^{\ell}(x_1,x_2) = \inf_{(x_1,x_2) \leq (y_1,y_2)} f(y_1,y_2) = f(x_1,1)
    \]
\end{example}

To obtain a safe bound on the fixed points, we progress the upper and lower bound cautiously by finding safe solutions to $x \leq f^u(x)$ and $x \geq f^{\ell}(x)$.
Since $f^u$ and $f^{\ell}$ are monotone, they have a greatest/least solution to their respective problems by the lattice's partial-order unlike $f$. In particular, the safest (greatest) solution to $x \leq f^u(x)$ is
\[
\bigsqcup \{x \in L \mid x \leq f^u(x)\} = \operatorname{gfp}(f^u)
\]
by the Knaster-Tarski fixed point theorem. The dual holds for $f^{\ell}$. Let 
\begin{align*}
    x^u &= \gfp{f^u}\\
    x^{\ell} &= \lfp{f^{\ell}}
\end{align*}

\begin{theorem}
    For all fixed points $\mu$ of $f$, $x^{\ell} \leq \mu \leq x^{u}$.
\end{theorem}

\begin{proof}
    Let $\mu$ be a fixed point of $f$. Then $\mu = f(\mu) \leq f^u(\mu)$, so $\mu \leq x^u$ by Knaster-Tarski. Similarly, $\mu = f(\mu) \geq f^{\ell}(\mu)$, so $x^{\ell} \leq \mu$ by Knaster-Tarski.
\end{proof}

Indeed, these are the \emph{best} achievable bounds in the sense that they are the tightest bounds obtainable by the tightest sound interval transformer for $f$ \cite{denecker2004}. Namely, $[x^{\ell},x^u]$ is AFT's well-founded approximation for $f$.

\begin{example}
    \label{example:step-2}
    We pick up where we left off in Example \ref{example:step-1}. By applying Kleene iteration on $f^u$ and $f^{\ell}$, we find their greatest and least fixed points respectively.
    \begin{align*}
        f^u(1,1) &= f(1,0) = (1,0.5)\\
        f^u(1,0.5) &= f(1,0) = (1,0.5)
    \end{align*}
    so $x^u = (1,0.5)$. To find $x^{\ell}$, we apply Kleene iteration on $f^{\ell}$:
    \begin{align*}
        f^{\ell}(0,0) = f(0,1) = (0,0)
    \end{align*}
    so $x^{\ell} = (0,0)$. We have reached the initial bound $$(0,0) \leq \mu \leq (1,0.5)$$ and as we shall soon see, this bound can be improved by refining $f^u$ and $f^{\ell}$.
\end{example}

\section{Iterative Approximation Refinement}
\label{sec:soundness-thms}

We can view an iterative refinement of the monotone envelopes as repeated application of a ``refinement'' function. That is, a function $F$ over intervals (or \emph{bounds}; they will be used interchangeably) such that $$F([\xubar{k-1},\xbar{k-1}]) = [\xubar{k},\xbar{k}]$$

We can generalize this function to the set of all intervals, not just those generated by $[\bot,\top]$ to show a useful result that states that if we begin at an arbitrary bound approximating an optimal solution, then we will stay in a sound state after repeated application of $F$.

We first define the function $F$ to be a function over the set of bounds $\B = \{\bbound \mid x^{\ell},x^u \in L \wedge x^{\ell} \leq x^u\}$ to be
\[
F(B) = [\lfpB{f_B^{\ell}}, \gfpB{f_B^u}]
\]
where $f_B^{\ell}$ and $f_B^u$ are defined by
\begin{align*}
    f_B^{\ell}(x) &= \inf_{x \leq y\leq x^u} f(y)\\
    f_B^u(x) &= \sup_{x^{\ell} \leq y \leq x} f(y)
\end{align*}
and
\begin{align*}
    \lfpB{g} &= \bigsqcap \{x \in L \mid x \geq x^{\ell} \land x \geq g(x)\}\\
    \gfpB{g} &= \bigsqcup \{x \in L \mid x \leq x^u \land x \leq g(x)\}
\end{align*}
are the Knaster-Tarski least/greatest fixed point definitions over a \emph{restricted} domain. We emphasize that the meet and join is in the lattice $L$, rather the interval's sub-lattice meet/join (for the cases of empty sets). Observe that $f^{\ell}_B$ and $f^u_B$ are monotone over $B$. They are a generalization of the earlier monotone envelopes. It is not hard to show that for the monotone functions $\varphi^{\ell}_B(x) = x^{\ell} \sqcup f^{\ell}_B(x)$ and $\varphi^u_B(x) = x^u \sqcap f^u_B(x)$, $\lfp{\varphi^{\ell}_B} = \lfpB{f^{\ell}_B}$ and $\gfp{\varphi^u_B} = \gfpB{f^u_B}$. Thus computing $\lfpB{f^{\ell}_B}$ and $\gfpB{f^u_B}$ reduces to computing the least/greatest fixed point of monotone functions.

\begin{example}
    Building on Examples \ref{example:step-1} and \ref{example:step-2}, we define the refined functions $\fbar{2}$ and $\fubar{2}$ ($f^u_B$ and $f^{\ell}_B$ for $B = [x^{\ell}_1,x^u_1]$). In Example \ref{example:step-2} we determined that $\xbar{1} = (1,0.5)$ and $\xubar{1} = (0,0)$. This gives
    \begin{align*}
        \fbar{2}(x_1,x_2) &= \sup_{(0,0)\leq (y_1,y_2) \leq (x_1,x_2)} f(y_1,y_2) = f(x_1,0)\\
        \fubar{2}(x_1,x_2) &= \inf_{(x_1,x_2) \leq (y_1,y_2) \leq (1,0.5)} f(y_1,y_2) = f(x_1,0.5)
    \end{align*}
    Here we observe that $\fbar{2}$ didn't improve upon $\fbar{1}$ since $\xubar{1}$ didn't improve upon $\xubar{0}$. On the other hand, the dual approximation $\fubar{2}$ improved since $\xbar{1}$ improved upon $\xbar{0}$.

    We can now look for $\xbar{2}$ and $\xubar{2}$ given the definitions of $\fbar{2}$ and $\fubar{2}$. The least/greatest fixed points can be found by using Kleene iteration on $\varphi^{\ell}_2$ and $\varphi^u_2$ as described earlier, but in this example the result is no different. Since $\fbar{2} = \fbar{1}$, we have $\xbar{2} = \xbar{1} = (1,0.5)$. For $\fubar{2}$ we find:
    \begin{align*}
        \fubar{2}(0,0) &= f(0,0.5) = (0.75, 0)\\
        \fubar{2}(0.75,0) &= f(0.75,0.5) = (0.75,0.375)\\
        \fubar{2}(0.75,0.375) &= f(0.75,0.5) = (0.75,0.375)
    \end{align*}
    and we see that the lower bound $\xubar{1} = (0,0)$ has been improved to $\xubar{2} = (0.75,0.375)$. This improvement was a result of the dual bound $\xbar{1}$ tightening the domain of $\fubar{2}$. Alas, we now see that our approximating bound for the fixed point has improved to
    \[
    (0.75,0.375) \leq \mu \leq (1,0.5)
    \]
\end{example}

In terms of AFT, $F$ is comparable to the stable operator for the ultimate approximation (which we will call the DMT operator after Denecker, Marek, and Truszczy\'{n}ski) \cite{denecker2004} with a small change: $F$ can map a valid bound to an invalid bound (not in $\B$). This difference comes from the fact that the stable operator \cite{denecker2004} takes the least/greatest fixed point of the monotone envelopes over the whole domain, while we decide to restrict the domain. At first glance, this change seems harmful since it takes away some of the nice properties enjoyed in AFT. We do this intentionally to introduce potentially \emph{bad} behavior when there is \emph{no} fixed point in the interval. This will allow us to sometimes determine when an interval is unsound, which will be useful for the pruning step of our later algorithm.

In particular, it is not always the case that $F(B) \in \B$ for an interval $B \in \B$. With the definition of $F$, this is only guaranteed to hold for some intervals. As motivated earlier, we are interested in bounds that are \emph{sound} in the same sense as we saw before.

\begin{definition}
    A bound $B\in \B$ is \emph{sound} (for $f$) if it contains any $\mu \in \fix{f}$.
\end{definition}

This definition states that a bound is sound if it contains some solution. As we will see, this keeps the solution contained after each iteration. We first ensure that $F(B)$ is indeed sound and contained in $\B$ for all sound $B \in \B$. For convenience, we will let $\B_{sound} = \{B \in \B \mid \text{$B$ is sound}\}$.

\begin{lemma}
    \label{lem:B-sound-F(B)-sound}
    For all $B \in \B_{sound}$, $F(B) \in \B_{sound}$.
\end{lemma}
\begin{proof}
    Since $B$ is sound, there exists a fixed point $\mu \in B$ of $f$. 
    Observe that $\mu = f(\mu) \leq f^u_B(\mu)$ and $\mu = f(\mu) \geq f^{\ell}_B(\mu)$ since $\mu \in B$ (and $f^{\ell}_B \leq f \leq f^u_B$ over $B$), so $\lfpB{f^{\ell}_B} \leq \mu \leq \gfpB{f^u_B}$ by definition.
\end{proof}

We now show that $F$ can only shrink bounds (provided they are sound), and directly after we show that if a bound is a sound approximation of any optimal solution, then it remains a sound approximation after applying $F$.

\begin{lemma}
    \label{lem:F-reductive}
    For all $B \in \B_{sound}$, $F(B) \subseteq B$.
\end{lemma}
\begin{proof}
    Let $B = \bbound$. It then suffices to show that $\lfpB{f^{\ell}_B} \geq x^{\ell}$ and $\gfpB{f^u_B} \leq x^u$. Since $B$ is sound, $\lfpB{f^{\ell}_B}$ is a meet of a non-empty set of elements with $x^{\ell}$ as a lower bound, so $x^{\ell} \leq \lfpB{f^{\ell}_B}$. A dual argument follows to show $\gfpB{f^u_B} \leq x^u$.
\end{proof}

These two lemmas will serve to be very useful in Section \ref{sec:bnb}. To illustrate how Lemma \ref{lem:B-sound-F(B)-sound} can be useful, the following example shows that if a bound is not sound, applying $F$ may take it to an invalid interval.

\begin{example}
    Consider the function $f$ from Example \ref{example:step-1}, starting from the bound $B = [(0,0),(0.2,0.2)]$, which contains no fixed points. Applying $F$ on this bound gives
    $$F([(0,0),(0.2,0.2)]) = [\lfp{\varphi^{\ell}_B},\gfp{\varphi^u_B}]$$ and Kleene iteration on these functions gives $\lfp{\varphi^{\ell}_B} = (1,1)$ and $\gfp{\varphi^u_B} = (0.2,0.1)$, and $[(1,1),(0.2,0.1)] \notin \B$, so by Lemma \ref{lem:B-sound-F(B)-sound} it is necessarily true that there are no fixed points in $[(0,0),(0.2,0.2)]$.
\end{example}

\begin{lemma}
    \label{lem:combine-lems}
    For all $B \in \B$, if $F(B) \in \B$ then $F(B) \subseteq B$.
\end{lemma}
\begin{proof}
    Let $B = \bbound$ and observe that $x^{\ell}$ is a lower bound for $\lfpB{f^{\ell}_B}$ (even if it is a meet of an empty set) and $x^u$ is an upper bound for $\gfpB{f^u_B}$ (even if it is a join of an empty set), so $F(B) \subseteq B$.
\end{proof}

\begin{lemma}
    \label{lem:F-monotone}
    If $B_1,B_2 \in \B_{sound}$ and $B_1 \subseteq B_2$, then $F(B_1) \subseteq F(B_2)$.
\end{lemma}
\begin{proof}
    Let $B_1,B_2 \in \B$ such that $B_1 \subseteq B_2$ and let $x \in B_1$ such that $x \geq f^{\ell}_{B_1}(x)$ (existence guaranteed by soundness of $B_1$). Then, since $B_1 \subseteq B_2$, $f^{\ell}_{B_1} \geq f^{\ell}_{B_2}$ over $B_1$ (meet of a smaller set of elements is larger), so $$x \geq f^{\ell}_{B_1}(x) \geq f^{\ell}_{B_2}(x)$$ and the set of pre-fixed points of $B_1$ is a subset of the set of pre-fixed points of $B_2$. Thus, $$\lfpS{f^{\ell}_{B_1}}{B_1} \geq \lfpS{f^{\ell}_{B_2}}{B_2}$$ Dually, we observe that $\gfpS{f^u_{B_1}}{B_1} \leq \gfpS{f^u_{B_2}}{B_2}$ holds. Finally, this shows that $F(B_1) \subseteq F(B_2)$.
\end{proof}

The following theorem essentially restates known monotone iteration convergence results. We state the theorem to emphasize that by starting at a bound that is not generated by $\{F^k([\bot,\top])\}$, we do not end in a bound worse than $X_{\infty} = \inf_{k \geq 0} F^k([\bot,\top])$ (the bound approached by starting at the interval of the entire lattice).

\begin{theorem}
    Let $B \in \B_{sound}$ and let $B_{\infty} = \inf_{k \geq 0} F^k(B)$. Then, $B_{\infty} \subseteq X_{\infty}$.
\end{theorem}
\begin{proof}
    Since $\{F^k(B)\}_{k \geq 0}$ is a descending chain by Lemma \ref{lem:F-reductive}, $B_{\infty}$ is the greatest lower bound of that chain. Similarly, $X_{\infty}$ is the greatest lower bound of the descending chain $\{F^k([\bot,\top])\}_{k \geq 0}$. By monotonicity of $F$ (in the sound intervals, Lemma \ref{lem:F-monotone}) and $B \subseteq [\bot,\top]$, we then have $F^k(B) \subseteq F^k([\bot,\top])$ for all $k \geq 0$. Then for all $k \geq 0$, $B_{\infty} \subseteq F^k(B) \subseteq F^k([\bot,\top])$. Then, $B_{\infty}$ is a lower bound of the chain $\{F^k([\bot,\top])\}_{k \geq 0}$, so $B_{\infty} \subseteq X_{\infty}$.
\end{proof}

These results have a useful consequence: we may begin our process at any interval, and applying the iterative refinement process will always take us to an interval that is at least as good as $X_{\infty}$. Furthermore, if we begin at a sound interval, the iterative refinement process will end in a sound interval. This can be thought of as a soundness property of $F$.

This allows us to interleave other methods and heuristics to forcibly shrink the bound (when possible). These heuristics can be thought of as a ``jump-start'' for $F$, and the refinement process defined by $F$ can be thought of as a search space restriction for the heuristics. The jump-start is especially useful when the bound is stuck at a fixed point of $F$. Recall that at each step of refinement, the upper and lower bounds of a sound interval \emph{always} remain above and below the solutions previously contained. By forcibly pushing either the upper or the lower bound in such a way they no longer bound some solutions, but still bound at least one (to remain sound), we can effectively jump-start the refinement process so it is no longer stuck, and hope for a snowball effect where the bounds continue improving each other once again. This can be visually depicted in Figure \ref{fig:stuck_xl}, where $\bbound$ is stuck since it must bound all solutions. In Figure \ref{fig:unstuck_xl}, $x^{\ell}$ is forcibly advanced (using some heuristic or a domain-specific sound advancement), which enables $x^u$ to continue improving.

In the Section \ref{sec:bnb}, we take advantage of this observation to create an effective algorithm for narrowing in on exact fixed points. It also makes note of a second useful observation: Lemmas \ref{lem:B-sound-F(B)-sound} and \ref{lem:F-reductive} show that if $F^k(B)$ is not a valid interval or it is not a subset of $B$, then it \emph{must} be the case that $B$ is not sound. We will use this fact to prune the search space and maintain a more manageable set of search intervals.

\begin{figure}[ht]
\centering
\begin{subfigure}[t]{.4\textwidth}
  \centering
  \includegraphics[width=.9\linewidth]{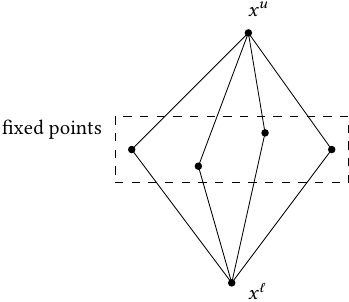}
  \caption{$x^{\ell}$ is stuck below many fixed points.}
  \label{fig:stuck_xl}
\end{subfigure}%
\hfill
\begin{subfigure}[t]{.47\textwidth}
  \centering
  \includegraphics[width=\linewidth]{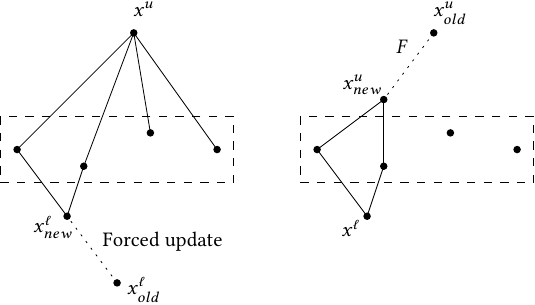}
  \caption{$x^{\ell}$ is forcibly advanced while remaining sound, allowing $x^u$ to progress.}
  \label{fig:unstuck_xl}
\end{subfigure}
\caption{Illustration of ``jump-starting'' a stuck refinement.}
\Description[Forcibly updating one bound makes the other bound improve]{In one diagram, the lower and upper bounds are stuck bounding many fixed points. In a second diagram, the lower bound is depicted being forcibly increased. In a third diagram, the upper bound is depicted being decreased by application of $F$.}
\end{figure}

\subsection{An Abstract Interpretation for the Approximation}

In many cases, computation of $f^{\ell}_B$ and $f^u_{B}$ is very expensive. It may require enumeration over an exponentially sized lattice (in the size of the problem input). While the examples in this paper are centered around instances where these functions are relatively cheap to compute, we present a new theorem, to the best of our knowledge, that enables an abstraction for sound non-monotone fixed point approximations. In this section, we will let $\leq_p$ be an ordering over pairs: $(x_1,y_1) \leq_p (x_2,y_2)$ iff $x_2 \leq x_1 \land y_1 \leq y_2$. This ordering represents the precision of an interval approximation. It is evident that for $B_1,B_2 \in \B$ represented as pairs of the lower and upper bounds, $B_1 \leq_p B_2$ iff $B_1 \subseteq B_2$. It can then be thought of as an extension of set inclusion to all pairs. We will think of $\B$ as being a subset of $L^2$ here.

We will refer to $F_f$ as the function $F$ constructed before but parameterized by the function $f$. Let $C$ and $A$ be complete lattices. Let $\alpha^u,\alpha^{\ell} : C \to A$ and $\gamma^u,\gamma^{\ell} : A \to C$ be monotone such that $\alpha^u \dashv \gamma^u$ and $\gamma^{\ell} \dashv \alpha^{\ell}$ are Galois insertions. We will also let, for any $B = \bbound \in C^2$, $\alpha(B) = [\alpha^{\ell}(x^{\ell}),\alpha^u(x^u)] \in A^2$ (and similarly for $\gamma$). Note that, as shown in the proof of the following theorem, $\alpha^{\ell} \leq \alpha^u$.

\begin{theorem}
    Let $f : C \to C$ and $f^{\sharp},f^{\flat} : A \to A$ satisfy $$f^{\flat} \circ \alpha^{\ell}\leq_A \alpha^{\ell} \circ f \qquad \alpha^u\circ f \leq_A f^{\sharp} \circ \alpha^u$$ and for all $B \in A^2$ let $F^{\sharp} : A^2 \to A^2$ be defined by $F^{\sharp}(B) = [\lfpB{(f^{\flat})^{\ell}_B}, \gfpB{(f^{\sharp})^u_B}]$. 
    Then for all $B \in C^2$, $$\alpha(F(B)) \leq_p F^{\sharp}(\alpha(B)).$$
\end{theorem}
\begin{proof}
    See appendix.
\end{proof}

This theorem gives a simple recipe for an abstract interpretation on the approximation operator $F$. The only requirement is to have a such $\alpha^{\ell},\gamma^{\ell}$ and $\alpha^u,\gamma^u$. The first pair can be thought of as an under-approximating abstraction while the second pair can be thought of as an over-approximating abstraction (in the standard abstract interpretation sense). Then, we may use the best transformers $f^{\flat} = \alpha^{\ell}f \gamma^{\ell}$ and $f^{\sharp} = \alpha^u f\gamma^u$, which are known to satisfy the inequalities required by the theorem. It is not hard to see that these best correct approximations (in the Cousot sense \cite{cousot1979}) also yield the best correct approximation $F^{\sharp}$ for $F$.

\begin{corollary}
    If $\alpha(F(B)) \in \B^A$, then $F^{\sharp}(\alpha(B)) \in \B^A$ $\alpha(F(B)) \subseteq F^{\sharp}(\alpha(B))$.
\end{corollary}

\section{A Branch-and-Bound for Tightening Approximations}
\label{sec:bnb}

Thus far, the described iterative refinement process has been very safe: it progresses monotonically only to a provably sound bound at each step. As previously hinted, we now look at methods for speculatively advancing the bound rather than cautiously advancing. This is achieved using a branch-and-bound style algorithm that functions using iterative refinement and pruning when a bound is provably unsound. In the context of this algorithm, we will let $IR(B)$ be the repeated iteration of $F$ on $B$. This can be any number of iterations, whether it is to a fixed point or not. The more iterations, the tighter the resulting bound (if $B$ is sound, by Lemma \ref{lem:F-reductive}).

\paragraph{Refine}
Suppose we have a set of sound bounds $Bounds$. For each $B \in Bounds$, let $B_1,B_2$ be a decomposition of $B$ so that $B = B_1\cup B_2$ and each $B_i$ is a valid bound ($B_i \in \B$). Then, run the iterative refinement algorithm on each $B_i$ to get $IR(B_i)$. By Lemma \ref{lem:B-sound-F(B)-sound}, all the fixed points that were contained in the intervals of $Bounds$ are still contained after this operation. 

\paragraph{Prune}
Let $IR(B_i) = \bbound$. If $x^{\ell} \not\leq x^u$, then by Lemma \ref{lem:B-sound-F(B)-sound}, it \emph{must} be the case that there are no fixed points of $f$ contained in $B_i$ (that is, $B_i$ is not sound). Similarly, if $IR(B_i) \not\subseteq B_i$, Lemma \ref{lem:F-reductive} proves that $B_i$ is unsound. In these cases, we will prune $B_i$ from the search space $Bounds$. Note that these conditions are sufficient to determine that $B_i$ is not sound, but not necessary.

\paragraph{Merge}
We then merge any ``adjacent'' bounds in the following sense: for any two bounds $A_1 = [a_1^{\ell},a_1^u] \in \B, A_2 = [a_2^{\ell},a_2^u] \in \B$, $A_1$ and $A_2$ are \emph{adjacent} if $A_1 \cup A_2 \in \B$ (where $\cup$ is typical set union rather than a join of intervals). Finally, we update $Bounds$ by joining all adjacent bounds (after pruning and iterative refinement), resulting in a set of refined, non-adjacent bounds.

After repeating this process to a fixed point (or stopping earlier), we end with a set of smaller bounds such that every fixed point $\mu$ of $f$ is contained in some $B \in Bounds$. Observe that this algorithm utilizes the soundness theorems. An early-stopping technique can be added by checking if the endpoints of a bound is a fixed point after iterative refinement. In this case, we do not need to continue the search, and we can simply return the fixed point. If we are interested in \emph{all} fixed points or in an \emph{optimal} fixed point (by some norm), then we may store the discovered fixed point, update the current interval, and continue the search.

In Algorithm \ref{alg:bnb}, we will let $$\phi(B) \equiv IR(B) \in \B \land IR(B) \subseteq B$$
If $\neg \phi(B)$, then Lemmas \ref{lem:B-sound-F(B)-sound} and \ref{lem:F-reductive} prove that $B$ is not sound. If at any point during the $IR$ process, these conditions fail, then we may choose to stop it immediately to save time. We also use a procedure $\textsc{Decompose}(B)$. Assume that for any $B \in \B$, $\textsc{Decompose}(B)$ correctly returns a pair of adjacent non-trivial (when possible) bounds $B_1,B_2 \in \B$ such that $B_1\cup B_2 = B$. The if-statement's guard on line 8 ensures that if $B$ is a singleton, then it will be finalized. We later show how such a procedure can easily be devised in the case of a lattice of subsets.

\begin{algorithm}
\caption{Sound Branch-and-Bound For Tightening Converged Approximations}
\label{alg:bnb}
\begin{algorithmic}[1]
  \STATE \textbf{input:} \text{a sound bound $B_I$}
  \STATE $Bounds = \{IR(B_I)\}$
  \STATE $Final = \{\}$
  \WHILE{$Bounds \neq \emptyset$}
    \STATE $NewBounds \coloneqq \{\}$
    \FOR{$B \in Bounds$}
        \STATE $B_1,B_2 \coloneqq \textsc{Decompose}(B)$
        \IF{$B_1 = B$ \textbf{ or } $B_2 = B$}
            \STATE $Final \coloneqq Final \cup \{B\}$
            \STATE \textbf{continue}
        \ENDIF
        \STATE $Keep = \{IR(B_i) \mid i \in \{1,2\} \land \phi(B_i)\}$
        \IF{$Keep$ contains two adjacent bounds}
            \STATE $Keep \coloneqq \bigcup Keep$
        \ENDIF
        \IF{$Keep = \{B\}$}
            \STATE $Final \coloneqq Final \cup \{B\}$
        \ELSE
            \STATE $NewBounds \coloneqq NewBounds \cup Keep$
        \ENDIF
    \ENDFOR
    \STATE $Bounds \coloneqq NewBounds$
  \ENDWHILE
  \STATE \textbf{return} $Final$
\end{algorithmic}
\end{algorithm}

\begin{lemma}
    \label{lem:bnb-invars}
    The following are invariants of the outer while loop in Algorithm \ref{alg:bnb}:
    \begin{enumerate}
        \item $Bounds \subseteq \B$.
        \item $\forall B \in Bounds$, $B \subseteq B_I$.
        \item $\forall B_1,B_2 \in Bounds$ such that $B_1 \neq B_2$, $B_1$ and $B_2$ are not adjacent.
    \end{enumerate}
\end{lemma}
\begin{proof}
    See appendix.
\end{proof}

\begin{theorem}
    \label{thm:termination}
    If $L$ satisfies both ACC and DCC, Algorithm \ref{alg:bnb} terminates in a finite number of iterations.
\end{theorem}
\begin{proof}
    See appendix.
\end{proof}

\begin{theorem}
    \label{thm:anytime-completeness}
    The following invariant holds: for all $\mu \in \fix{f}$ in $B_I$, there exists a bound $B \in Final \cup Bounds$ such that $\mu \in B$.
\end{theorem}
\begin{proof}
    See appendix.
\end{proof}

\begin{corollary}
    Upon termination, all fixed points of $f$ are contained in at least one interval $B \in Final$.
\end{corollary}

\begin{corollary}
    If the intervals produced by \textsc{Decompose} are disjoint, then the invariant from Theorem \ref{thm:anytime-completeness} strengthens: for all fixed points $\mu \in B_I$ of $f$, there exists a \emph{unique} $B \in Final \cup Bounds$.
\end{corollary}

We substantiate the improved tightness of the resulting approximating bounds in two ways. The first is by observing that the proof of invariant (2) of Lemma \ref{lem:bnb-invars} proves $\forall B \in Bounds$, $B \subseteq IR(B_I)$. The following theorem shows the second way, proving that the total size of the search space is no larger than the AFT best approximation $IR(B_I)$. We use $|\cdot |$ to mean set cardinality for interval sets as well.

\begin{theorem}
    \label{thm:search-space-smaller}
    Let $S(Final \cup Bounds) = \sum_{B \in Final\cup Bounds} |B|$. Then, $S(Final \cup Bounds) \leq |IR(B_I)|$ is an invariant of the outer loop of Algorithm \ref{alg:bnb} if the intervals produced by \textsc{Decompose} are disjoint.
\end{theorem}
\begin{proof}
    This comes immediately from the fact that the bounds in $Final\cup Bounds$ are disjoint and contained in $IR(B_I)$. Thus the size of the (set) union of the intervals in $Final \cup Bounds$ is equal to $S(Final\cup Bounds)$, and the union is a subset of $IR(B_I)$, so $S(Final\cup Bounds) \leq |IR(B_I)|$.
\end{proof}

To illustrate an implementation of \textsc{Decompose} whose produced sub-intervals are disjoint, let $L = 2^S$ for some set $S$ ordered by subset inclusion. Then, if $[A,B] \in \B$ (with non-zero height), arbitrarily choose $x \in B \setminus A$. It is evident that $[A, B\setminus \{a\}]$ and $[A \cup \{a\}, B]$ are adjacent and disjoint. Observe that the effectiveness of Algorithm \ref{alg:bnb} can be dependent on the implementation of \textsc{Decompose}. If the decomposition happens to be good (in the sense that it enables the iterative refinement to greatly reduce the decomposed bounds), then the algorithm will be very effective. One may decide to implement a context-aware version of \textsc{Decompose} that can learn from decompositions that did not previously work. This would require \emph{not} finalizing refined bounds that join back to their parent bound.

As it stands, the proof of Theorem \ref{thm:termination} tells us that the outer loop runs for at most $\ell(B_I)$ iterations. In the case of $2^{S}$, this states that it runs for at most $|S|$ iterations. The issue is that $|Bounds|$ may be exponentially large, so the inner loop can take exponential time. To remedy this, we consider a \emph{budget condition} modification for bounding $|Bounds|$ to polynomial size. For bounds $B_1 = [x^{\ell}_1,x^u_1] \in \B$ and $B_2 = [x^{\ell}_2, x^u_2] \in B$, we will let $\operatorname{hull}(B_1,B_2) = [x^{\ell}_1 \sqcap x^{\ell}_2, x^u_1 \sqcup x^u_2]$ (a join under the precision ordering).

\paragraph{Budget Condition}
We can modify Algorithm \ref{alg:bnb} as follows: let $n = \ell(B_I)$ and $K = \operatorname{poly}(n)$ (and non-zero). Whenever a bound is added to $NewBounds$ and $|NewBounds| > K$, choose a pair $B_1,B_2 \in NewBounds$ that minimizes $$|\operatorname{hull}(B_1,B_2)| - (|B_1| + |B_2|)$$ and replace $B_1,B_2$ with $\operatorname{hull}(B_1,B_2)$. Note that pruning its subsets is not necessary since if $\operatorname{hull}(B_1,B_2)$ is added and there exists $C \subseteq \operatorname{hull}(B_1,B_2)$, then either $B_1,C$ or $B_2,C$ would have been chosen to minimize the delta instead. This heuristic keeps the cumulative size of the intervals as small as possible while bounding $|Bounds|$ to polynomial size. With this budget condition, we may safely assume that $|Bounds| \leq K$ is invariant. Observe that adding this condition breaks the guarantee of termination, even with finite height lattices. By Theorem \ref{thm:anytime-completeness}, we can choose to stop at any point while maintaining soundness (in $Final \cup Bounds$). We will let $T$ be a chosen number of iterations to perform before forcibly terminating. This modification upholds the invariant of Theorem \ref{thm:anytime-completeness}, and is forced to terminate in finite time given that $IR$ terminates in finite time. This is described by the following theorem and the discussion that follows. 

\begin{theorem}
    \label{thm:budget-condition}
    With the budget condition, the number of $IR$ calls in Algorithm \ref{alg:bnb} is no greater than $2KT+1 = \operatorname{poly}(n)\cdot T$.
\end{theorem}
\begin{proof}
    The budget condition guarantees that $|Bounds| \leq K$ is invariant. The only time $IR$ is called is (1) at initialization, and (2) in the inner loop. Initialization contributes one $IR$ call, and each iteration of the inner loop contributes at most $2$ calls. Since $|Bounds| \leq K$, the inner loop performs at most $K$ iterations, so the inner loop contributes at most $2K$ calls. Finally, since the outer loop iterates no more than $T$ times, it contributes no more than $2KT$ times, so there are a total of $2KT + 1$ calls to $IR$.
\end{proof}

This theorem shows that, the budget condition modification, if a call to $IR$ takes $\operatorname{poly}(n)$ time (whether by structural properties of $f$, lattice abstractions, or restricted iterations of $F$), Algorithm \ref{alg:bnb} takes $\operatorname{poly}(n)\cdot T$ time. This makes note of the fact that choosing a $B_1,B_2$ to replace by their hull is a polynomial operation in $K$ (and thus in $n$). Additionally, if $T$ is chosen to be $\operatorname{poly}(n)$, then the algorithm takes $\operatorname{poly}(n)$ time. The trade off is that the produced bounds will be larger and Theorem \ref{thm:search-space-smaller} may not always hold. This trade off is, in most cases, worth it (as is later illustrated).

\subsection{An Example With Answer Set Programming}
\label{sec:asp}

Consider the problem of finding a minimum size stable model for a logic program $P$. By definition, the stable models are exactly the subsets $M$ of the Herbrand base of $P$ such that $M = S_P(M)$. It was already established that $S_P$ is not monotone, but rather it is anti-monotone. In this section, we will see how anti-monotonicity makes approximating the fixed points of $S_P$ an easy computation, and demonstrate an application of Algorithm \ref{alg:bnb} on a simple non-stratified logic program.

For $k \geq 1$, let $\SPubar{k},\SPbar{k}$ be the resulting envelopes after $k$ applications of $F$ with the function $S_P$, and let $F^k([\emptyset,H_P]) = [\xubar{k},\xbar{k}]$.
\begin{align*}
    \SPubar{k}(M) &= \inf_{M \subseteq A \subseteq \xbar{k-1}}S_P(A) = S_P(\xbar{k-1})\\
    \SPbar{k}(M) &= \sup_{\xubar{k-1} \subseteq A \subseteq M}S_P(A) = S_P(\xubar{k-1})
\end{align*}
by anti-monotonicity of $S_P(A)$. Thus, computing each of the approximating functions is a matter of a single application of $S_P$. Furthermore, since $\SPubar{k}$ and $\SPbar{k}$ are constant functions, their lfp/gfp-equivalent functions admit a single known fixed point. Thus computing $\lfpB{\SPubar{k}}$ and $\gfpB{\SPbar{k}}$ requires only a single application of $S_P$ for each, which is an easy computation. Specifically, for $k > 0$
\begin{align*}
    \lfpB{\SPubar{k}} &= \xubar{k-1} \cup S_P(\xbar{k-1})\\
    \gfpB{\SPbar{k}} &= \xbar{k-1} \cap S_P(\xubar{k-1})
\end{align*}

This approximation of the stable models is equivalently presented as an example in \cite{denecker2004}. We will show that Algorithm \ref{alg:bnb} can be used to narrow in on the \emph{exact} stable models by taking advantage of a known fixed point theorem for anti-monotone functions.

For an anti-monotone function $f$, if $\mu$ is the least fixed point of the monotone function $f^2$, then $\mu$ and $\nu = f(\mu)$ are the \emph{dual} fixed points (sometimes called the \emph{extreme oscillating pair}) of $f$ where $f(\nu) = \mu$ and $f(\mu) = \nu$ \cite{gelder1993}. Furthermore, $\mu$ is a lower bound for the fixed points of $f$ and $\nu$ is an upper bound for the fixed points of $f$.

It is not hard to see that the iterative refinement process $IR$ for an anti-monotone function beginning at $[\bot,\top]$ converges (in a finite height lattice) at the bound $[\mu,\nu]$, and thus in the case of $S_P$, the converged upper bound is a set of atoms that are true in \emph{some} stable model of $P$ and the converged lower bound is a set of atoms that is true in \emph{all} stable models of $P$ (relating to the \emph{well-founded semantics}). That is, the lower bound is \emph{certain} and the upper bound is \emph{possible} among the stable models.

We will now illustrate this with an example \emph{non-stratified} logic program $P$.

\begin{align*}
    p &\gets \neg q \qquad r \gets p\\
    q &\gets \neg p \qquad s \gets q
\end{align*}
Note that the stable models of $P$ are $\{p,r\}$ and $\{q,s\}$. Their intersection is empty and their union is the full Herbrand base, so we expect the iterative refinement to get stuck at this unimproved interval. We begin with the bounds $\xubar{0} = \emptyset$ and $\xbar{0} = H_P$ (the Herbrand base of $P$). Then as derived earlier,
\begin{align*}
    \xubar{1} &= S_P(\xbar{0}) = \emptyset\\
    \xbar{1} &= S_P(\xubar{0}) = H_P
\end{align*}
and indeed, $[\xubar{1},\xbar{1}]$ (the AFT canonical best approximation) is stuck bounding the entire space. To close in on the stable models, we apply Algorithm \ref{alg:bnb} with the input bound $H_P$. As already shown, $IR(H_P) = [\emptyset,H_P]$. Now this bound will be decomposed. To do so, we may use the method for set lattices described earlier. We first choose an arbitrary element in $H_P \setminus \emptyset$, say $p$, and split the interval into $B_1 = [\emptyset, H_P\setminus \{p\}]$ and $B_2 = [\emptyset \cup \{p\}, H_P]$. Then we apply iterative refinement on each of these bounds, beginning with the first:
$$F([\emptyset, H_P\setminus \{p\}]) = [\{q,s\},\{q,r,s\}]$$
so $\phi(B_1)$ holds and $IR(B_1)$ is kept. For $IR(B_2)$, we observe:
\[
F([\{p\}, H_P]) = [\{p\}, \{p,r\}]
\]
and at this point, it is evident that these bounds will be split to singletons and the algorithm will terminate. Discarding the singleton intervals before finalizing that do not have a stable model leaves us with $\{p,r\}$ and $\{q,s\}$, the exact stable models.

\section{Application to Speculative Program Analyses}

Finally, we have reached a point where all prior algorithms, results, and examples can work together in an application to a speculative program analysis. The problem faced in this section, when observed from a certain angle, reflects the kind of reasoning seen in stable model semantics. A standard program analysis works by soundly approximating the program state, whereas a speculative program analysis may make potentially incorrect assumptions about the program in an attempt to improve either precision or speed, and may need to retract assumptions if they are proven unsound by the analysis. A set of assumptions are \emph{stable} if they reach a fixed point in the program analysis. That is, under a stable set of assumptions $\sigma$, the least fixed point of the forward abstract transfer function under the assumptions of $\sigma$ does not disprove $\sigma$. With this framing, the parallelism to stable model semantics becomes more evident. Since they share the same flavor, we will see that the applicability of the earlier methods also apply nicely here.

Let $P$ be a program with (concrete) collecting semantics $\mathcal{S}$ and a monotone, Scott-continuous forward transformer $F_P : \mathcal{S} \to \mathcal{S}$. Let $(D,\sqsubseteq)$ be a complete lattice and $\alpha \dashv \gamma$ be a Galois connection between $\mathcal{S}$ and $D$. In a typical fashion, let $F_P^{\sharp} : D \to D$ be the standard sound abstract transfer function (that is both monotone and Scott-continuous). We will assume a finite set $\mathcal{A}$ of \emph{speculations} that will encode assertions about program states. For instance, an element in $\mathcal{A}$ may encode the assertion ``at location $\ell$, $y > 0$.'' Given an assertion $a \in \mathcal{A}$, we let $\Pi_a : D \to D$ be an abstract projector that \emph{strengthens} a given abstract state using the assumption that $a$ is true. We may reasonably require $\Pi_a$ to be monotone (if $S^{\sharp}_1 \sqsubseteq S_2^{\sharp}$, then applying the same strengthening to both preserves the ordering), reductive ($\Pi_a(S^{\sharp}) \sqsubseteq S^{\sharp}$), and idempotent ($\Pi_a(\Pi_a(S^{\sharp})) = \Pi_a(S^{\sharp})$). For a set of assertions $\sigma \subseteq \mathcal{A}$, we conveniently let $\Pi_{\sigma}$ be the pointwise meet of the abstract projectors $\{\Pi_a \mid a \in \sigma\}$.

We also let for $a \in \mathcal{A}$, $\ok{a} : D \to \{\textbf{true}, \textbf{false}, \textbf{unknown}\}$ be a function that determines $a$'s compatibility with the given abstract state. In the case of static analyses, this may be defined by
\[
\ok{a}(S^{\sharp}) = \begin{cases}
    \textbf{false} & \text{if $S^{\sharp}$ entails $\neg a$}\\
    \textbf{true} & \text{if $S^{\sharp}$ proves $a$}\\
    \textbf{unknown} & \text{otherwise}\\
\end{cases}
\]
In a dynamic analysis, this may be defined by $\ok{a}(S^{\sharp}) = \textbf{false}$ if any runtime counter-example has occurred, and \textbf{unknown} otherwise (in general a two-valued definition determining whether or not a counterexample has been seen). For any set of assertions $\sigma \subseteq \mathcal{A}$, the $\sigma$-constrained abstract transformer is $$F^{\sharp}_{P,\sigma} = \Pi_\sigma \circ F^{\sharp}_P$$ and $F^{\sharp}_{P,\sigma}$ is monotone and Scott-continuous for any fixed $\sigma$. Observe that $F^{\sharp}_{P,\sigma}$ is \emph{anti-monotone} in $\sigma$ (a smaller set of assertions leads to a larger abstract state in the precision ordering). For fixed $\sigma$, the standard abstract analysis gives $$A(\sigma) = \lfp{F^{\sharp}_{P,\sigma}}$$

To extract a set of valid assumptions, we first define the set $\mathrm{Safe}(S^{\sharp}) = \{a \in \mathcal{A} \mid \ok{a}(S^{\sharp}) \neq \textbf{false}\}$. That is, $\mathrm{Safe}(S^{\sharp})$ is the set of assumptions that are ``safe to keep for now'' (those that the abstraction can't refute). This definition can be adjusted to only keep self-justified assumptions by instead taking the assertions $a \in \mathcal{A}$ where $\ok{a}(S^{\sharp}) = \textbf{true}$. Finally, an operator for updating the current set of assumptions $\sigma$ can be defined as $$\Phi(\sigma) = \mathrm{Safe}(A(\sigma))$$

An intuition for this operator is that we begin with a set of assumptions $\sigma$, analyze the program under these assumptions, and then retract any such speculations that are falsified by $A(\sigma)$. Furthermore, the stable sets of assumptions are exactly the fixed points of $\Phi$, since $\sigma = \Phi(\sigma)$ is precisely consistent with the speculative analysis $A(\sigma)$.

\begin{lemma}
    $\mathrm{Safe} : D \to 2^{\mathcal{A}}$ is monotone.
\end{lemma}
\begin{proof}
    Let $S^{\sharp}_1 \sqsubseteq S^{\sharp}_2$, and let $\mathrm{Ref}_a(S^{\sharp}) \equiv \ok{a}(S^{\sharp}) = \textbf{false}$ be a predicate that is true iff assertion $a$ is refuted by $S^{\sharp}$. Then, $\mathrm{Ref}_a(S^{\sharp}_2) \Rightarrow \mathrm{Ref}_a(S^{\sharp}_1)$ since a more precise abstract state can refute more assertions. Then taking complements in $2^{\mathcal{A}}$ gives $\mathrm{Safe}(S^{\sharp}_1) \subseteq \mathrm{Safe}(S^{\sharp}_2)$.
\end{proof}

From this lemma and the anti-monotonicity of $A$ (due to the anti-monotonicity of $F^{\sharp}_{P,\sigma}$ w.r.t. $\sigma$), the operator $\Phi$ is anti-monotone. As shown in the previous section, this simplifies the iterative refinement to alternative fixed point iteration of an anti-monotone function, promoting computational efficiency. By keeping track of the analysis results throughout an application of Algorithm \ref{alg:bnb}, it terminates with the stable sets of assumptions (either exact or approximating bounds) \emph{along} with their resulting analysis. This can be beneficial for choosing an analysis that maximizes an objective such as precision (which in the case of precision, must be maximized by at least one of $\Phi$'s fixed points). In the particular case of dynamic analyses, counter-examples may arrive in an online fashion. If such a counter-example arrives for the current set of (previously stable) assumptions, it would be useful to have a more conservative fixed point (stable set of assumptions) to fall back on. The DMT operator alone is also useful (and cheaper to compute) in this case, since the resulting lower bound will be a set of \emph{safe} assumptions that appear in every stable set of assertions. The upper bound will then be a set of assumptions that are safe in \emph{some} stable set of assertions. This approximation, as discussed earlier, may be coarse, in which case it can be fed to Algorithm \ref{alg:bnb} for improvement.

The issue with Algorithm \ref{alg:bnb} in this application is that it requires running many analyses, which can be simply infeasible for large programs. To achieve speed-ups in this regard, memoizing basic transfer results $F^{\sharp}_P(\cdot)$ and reusing them in future computations of $A$ avoids lots of redundant computation. A lazy evaluation may then be done: if $A(\sigma)$ has been computed and we need $A(\sigma')$ where $\sigma \subseteq \sigma'$, simply take the results of $A(\sigma)$ and revisit only nodes affected by $\sigma' \setminus \sigma$ (and if $\sigma' \subseteq \sigma$, the same lazy evaluation can be done by revisiting only nodes affected by retracting $\sigma \setminus \sigma'$). Since most nodes won't be affected by one assertion (they are typically targeted to one small part of the program), this lazy evaluation will be fast and Algorithm \ref{alg:bnb} will become feasible to compute. The exact details of such a lazy evaluation are outside the scope of this paper. 

In the case of optimistic hybrid analyses \cite{devecsery2018}, we likely care more about the speed and resource usage of the dynamic analysis than the static one. In this case, finding stable sets of assumptions for the static analyzer does not affect the speed of the dynamic analysis. In fact, instead of sampling many executions to find likely invariants as prescribed in \cite{devecsery2018}, this technique can instead be used to find the likely invariants in a more predictable way with guarantees.

\begin{theorem}
    All stable sets of assumptions $\sigma \subseteq \mathcal{A}$ are contained in some interval in $Final$ upon termination of Algorithm \ref{alg:bnb}.
\end{theorem}

As a final note, we consider the relationship between a stable set of assumptions $\sigma$ and its corresponding analysis $A(\sigma)$, and how it can be changed. In a similar fashion as the assume-guarantee paradigm, $A(\sigma)$ is not necessarily sound with respect to the concrete program, but rather it is sound \emph{conditional} on $\sigma$. That is, the result is valid provided the speculations hold. If we instead define $\mathrm{Safe}_{\forall}(S^{\sharp}) = \{a \in \mathcal{A} \mid \ok{a}(S^{\sharp}) = \textbf{true}\}$ and $\Phi(\sigma) = \mathrm{Safe}_{\forall}(A(\sigma))$, then its fixed points contain only assumptions that have been \emph{proven} by the abstract facts. Then, $\Pi_{\sigma}$ only removes unreachable states, yielding $\alpha \circ F_P \sqsubseteq \Pi_{\sigma} \circ F^{\sharp}_P \circ \alpha = F^{\sharp}_{P,\sigma} \circ \alpha$. Thus with this definition of $\Phi$, the analysis is sound.

Given the circumstances, either an unsound or a sound analysis may be suitable. For example, the previously mentioned optimistic hybrid analysis takes advantage of the extra precision provided by unsound analyses (in hopes that it is validated). However, in a purely static setting unsoundness can lead to incorrect optimizations.

\section{Application as a Pre-Processing Step}

Theorems \ref{thm:anytime-completeness} and \ref{thm:search-space-smaller} are helpful for two reasons: (1) the final result provides a smaller search space containing the fixed points, which can be used in pre-processing, and (2) it shows that the algorithm is \emph{anytime} (and $IR(\cdot)$ is, by its definition, anytime), so it can be interrupted to have minimal impact on the total computation time. To illustrate this, we describe how it can be implemented as a pre-processing step for an ASP solver. There are a couple ways to do this.

First, suppose we have a bound $[A,B]$ on some stable models. One method is to partially evaluate the program using $A$ and $B$ by substituting $\textbf{true}$ for atoms in $A$ and $\textbf{false}$ for atoms \emph{not} in $B$ since all atoms in $A$ are in the stable models in $[A,B]$, and all atoms not in $B$ are not in those stable models. Then, simply drop the rules whose bodies become false, remove satisfied body literals, and simplify constraints/aggregates. This yields a much smaller, easier to compute program that an ASP solver can very quickly process. Having bounds on the stable models (like what is achieved by Algorithm \ref{alg:bnb}) can then significantly reduce the workload of an ASP solver.

This is illustrated in the example from Figures \ref{prog:3color} and \ref{prog:partial-eval}. Figure \ref{prog:3color} is the input logic program, and Figure \ref{prog:partial-eval} is its partial evaluation where the lower bound $A$ has atoms $red(1)$ and $green(4)$ and the upper bound $B$ does \emph{not} contain the atom $blue(2)$. While the partially evaluated program looks syntactically larger, its \emph{grounded program} (the program resulting from grounding all variables) is much smaller than the original program, as expected.

\begin{figure}[t]
  \centering
  \begin{minipage}{0.9\textwidth}
  \begin{lstlisting}
node(1..6).
edge(1,2). edge(1,3). edge(2,3). 
edge(3,4). edge(4,5). edge(5,6). edge(3,6).

red(N)  :- node(N), not blue(N), not green(N).
blue(N) :- node(N), not red(N), not green(N).
green(N) :- node(N), not red(N), not blue(N).

:- edge(N,M), red(N), red(M), N < M.
:- edge(N,M), blue(N), blue(M), N < M.
:- edge(N,M), green(N), green(M), N < M.
  \end{lstlisting}
  \end{minipage}
  \caption{A Logic Program for 3-Coloring}
  \label{prog:3color}
  \Description[Rules are described for encoding a 3-coloring problem]{The logic program contains rules encoding a graph, and rules enforcing the constraints of the 3-coloring problem on all vertices.}
\end{figure}

\begin{figure}[t]
  \centering
  \begin{minipage}{0.9\textwidth}
  \begin{lstlisting}
node(1..6).
edge(1,2). edge(1,3). edge(2,3). 
edge(3,4). edge(4,5). edge(5,6). edge(3,6).

% Truths from A
red(1). green(4).

% False atoms from H_P - B
:- blue(2)

% Facts derived from partial evaluation
:- red(2).  % from edge(1,2) and red(1)
:- red(3).  % from edge(1,3) and red(1)
:- green(3) % from edge(3,4) and green(4)
:- green(5) % from edge(4,5) and green(4)
green(2).   % not red(2) & not blue(2)
blue(3).    % not red(3) & not green(3)
:- blue(6)  % from edge(3,6) and blue(3)

% Remaining generators: only nodes 5 and 6 are undecided
red(5) :- not blue(5).
blue(5) :- not red(5).

red(6) :- not green(6).
green(6) :- not red(6).

:- red(5), red(6).
  \end{lstlisting}
  \end{minipage}
  \caption{Partially Evaluated Program with $\{red(1),green(4)\} \subseteq A$ and $blue(2) \notin B$}
  \label{prog:partial-eval}
  \Description[The complexity of the 3-coloring logic program reduced]{A partial evaluation of the logic program for 3-coloring is applied with bounds stating that node 1 is red and node 4 is green (lower bound), and node 2 is not blue (upper bound).}
\end{figure}

On the other hand, most ASP solvers support assumptions (as unit literals), which eliminates the need for the overhead of partially evaluating the program. Simply assume $a$ for all $a \in A$ and $\neg x$ for all $x \notin B$. Partial evaluation done manually still has its merits though: suppose we have bounds $[A,B]$ and $[A',B']$ where $A \subseteq A'$ (or similarly in the case $B' \subseteq B$). Then, instead of feeding the ASP solver the assumptions from $A$ and then separately feeding the assumptions from $A'$, we may partially evaluate the program by setting $A$ to true, and then store the partial evaluation $P'_A$. Then, we may later perform a partial evaluation on $P'_A$ setting the atoms $A' \setminus A$ to true. This compositional evaluation removes lots of redundant computation.

That being said, Theorem \ref{thm:anytime-completeness} gives that all the bounds in $Final \cup Bounds$ at any point contain all the fixed points, and we also know that these bounds are non-adjacent. These bounds can then be used as described earlier to speed up the computation of the ASP solver. This is seen in the diagram depicted in Figure \ref{fig:preproc}. Theorem \ref{thm:budget-condition} guarantees that applying the budget condition modification to Algorithm \ref{alg:bnb} restricts $|Bounds|$ to $\operatorname{poly}(n)$, thus refraining the process from making exponentially many calls to the ASP solver.

\begin{figure}[ht]
\centering
  \centering
  \includegraphics[width=.9\linewidth]{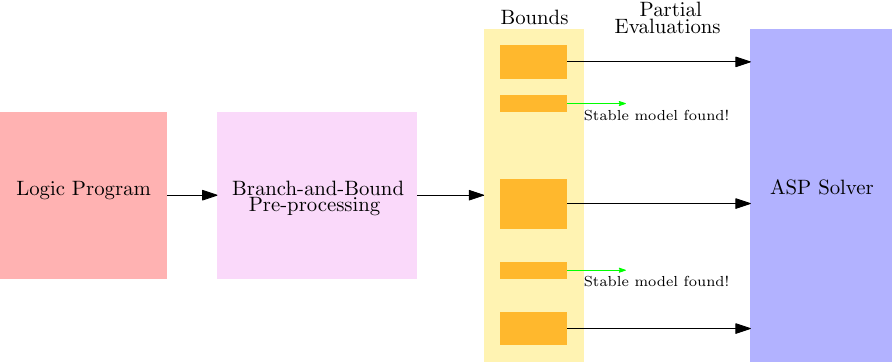}
  \caption{Insertion of a pre-processing step that bounds small subspaces containing the stable models for an ASP solver.}
  \label{fig:preproc}
  \Description[A diagram showing branch-and-bound as a pre-processing step]{The diagram depicts a logic program fed as input to the branch-and-bound algorithm, which outputs many small intervals, some of which are exact stable models. These intervals are then depicted as being fed into an ASP solver as assumptions}
\end{figure}

\section{Conclusion}

We present an AFT-oriented paradigm for modeling stability in non-monotonic processes for PL algorithms that utilize backtracking, revision, and other kinds of non-monotonic behavior. To solve these non-monotone fixed point problems in practical industrial settings, we (1) provide an abstract interpretation for the fixed point approximations, (2) relax a known AFT operator and propose an anytime branch-and-bound algorithm that takes advantage of its structural properties for pruning, and (3) show that it achieves tighter bounds than those guaranteed in AFT. We demonstrate the applicability of these methods for speculative program analysis and ASP, where we also describe its effectiveness at pre-processing for such problems.


\bibliographystyle{ACM-Reference-Format}
\bibliography{doc/refs}

\appendix

\paragraph{Proof of Theorem 5.9}
\begin{proof}
    We begin by proving that $\alpha^{\ell}(\lfpB{f^{\ell}_B}) \geq \lfpS{(f^{\flat})^{\ell}_{\alpha(B)}}{\alpha(B)}$ for all $B =\bbound \in \B^C$ which is equivalent to proving that $$\alpha^{\ell}(\lfp{\varphi^{\ell}_B}) \geq \lfp{(\varphi^{\flat})^{\ell}_{\alpha(B)}}$$ where $$\varphi^{\ell}_B(x) = x^{\ell} \sqcup_C f^{\ell}_B(x) \qquad (\varphi^{\flat})^{\ell}_{\alpha(B)} = \alpha^{\ell}(x^{\ell}) \sqcup_A (f^{\flat})^{\ell}_{\alpha(B)}(x)$$

    To show this, we first show that $\alpha^{\ell} \circ \varphi^{\ell}_B \geq (\varphi^{\flat})^{\ell}_{\alpha(B)} \circ \alpha^{\ell}$. Let $x \in C$ and observe
    \begin{align*}
        \alpha^{\ell}(\varphi^{\ell}_B(x)) = \alpha^{\ell}(x^{\ell} \sqcup_C f^{\ell}_B(x)) = \alpha^{\ell}(x^{\ell}) \sqcup_A \alpha^{\ell}(f^{\ell}_B(x))
    \end{align*}
    so it suffices to show that $\alpha^{\ell}(f^{\ell}_B(x)) \geq (f^{\flat})^{\ell}_{\alpha(B)}(\alpha^{\ell}(x))$.

    \begin{align*}
        \alpha^{\ell}(f^{\ell}_B(x))& = \alpha^{\ell}\left(\inf_{x \leq y \leq x^u} f(y)\right) = \inf_{x \leq y \leq x^u} \alpha^{\ell}(f(y)) \\&\geq \inf_{x \leq y \leq x^u} f^{\flat}(\alpha^{\ell}(y)) = \inf_{\alpha^{\ell}(x) \leq z \leq \alpha^{\ell}(x^u)} f^{\flat}(\alpha^{\ell}(y)) \\&\geq (f^{\flat})^{\ell}_{\alpha(B)}(\alpha^{\ell}(x))
    \end{align*}
    by $\alpha^{\ell} \leq \alpha^u$ ($\alpha^{\ell}(x) \leq \alpha^u(x) \Leftrightarrow x \leq \gamma(\alpha^u(x))$ which is always true). Thus $\alpha^{\ell} \circ f^{\ell}_B \geq (f^{\flat})^{\ell}_{\alpha(B)}\circ \alpha^{\ell}$, so $\alpha^{\ell} \circ \varphi^{\ell}_B \geq (\varphi^{\flat})^{\ell}_{\alpha(B)} \circ \alpha^{\ell}$.

    Observe that for any pre-fixed point $x \geq \varphi^{\ell}_B(x)$, we have $\alpha^{\ell}(x) \geq \alpha^{\ell}(\varphi^{\ell}_{B}(x)) \geq (\varphi^{\flat})^{\ell}_{\alpha(B)}(\alpha^{\ell}(x))$, so $\alpha^{\ell}$ maps the set of pre-fixed points of $\varphi^{\ell}_{\alpha(B)}$ to a subset of the pre-fixed points of $(\varphi^{\flat})^{\ell}_{\alpha(B)}$. Since $\alpha^{\ell}$ preserves meets, we then have
    \begin{align*}
        \alpha^{\ell}(\lfp{\varphi^{\ell}_B}) &= \alpha^{\ell}\left(\inf \{x \in C \mid x \geq \varphi^{\ell}_B(x)\}\right) = \inf \alpha^{\ell}\left(\{x \in C \mid x \geq \varphi^{\ell}_B(x)\}\right) \\&\geq \inf \{x \in A \mid x \geq (\varphi^{\flat})^{\ell}_{\alpha(B)}(\alpha^{\ell}(x))\} = \lfp{(\varphi^{\flat})^{\ell}_{\alpha(B)}}
    \end{align*}

    A dual argument shows that $\alpha^u(\gfp{\varphi^u_B}) \leq \gfp{(\varphi^{\sharp})^u_{\alpha(B)}}$ where
    $$\varphi^{u}_B(x) = x^{u} \sqcup_C f^{u}_B(x) \qquad (\varphi^{\sharp})^u_{\alpha(B)} = \alpha^u(x^u) \sqcup_A (f^{\sharp})^u_{\alpha(B)}(x)$$

    and therefore, $\alpha(F(B)) \leq_p F^{\sharp}(\alpha(B))$ by the definition of $F^{\sharp}$.
\end{proof}

\paragraph{Proof of Lemma 6.1}

\begin{proof}
    (1) Upon entering the outer loop, $Bounds$ is a singleton set $\{IR(B_I)\}$ which is in $\B$ (and sound) by Lemma 5.3 since $B_I$ is sound. Inside the loop, the only line where $Bounds$ is updated is line 22, where it is set to $NewBounds$. The only line where $NewBounds$ is updated is line 19, where it adds the elements of $Keep$, which contains either one or two elements. If $Keep$ contains only one element, then it is either the union of two adjacent intervals $IR(B_1),IR(B_2)$ or it is $IR(B_i)$ for either $B_1$ or $B_2$. In the former case, it must be the case that $IR(B_1)\cup IR(B_2) \in \B$ by definition of adjacency. In the latter case, $\phi(B_i)$ must be true, so $IR(B_i) \in \B$. In the case that $Keep$ contains two elements, then they must be $IR(B_1)$ and $IR(B_2)$, and they are not adjacent. Furthermore, it must be the case that $\phi(B_1)$ and $\phi(B_2)$ are true, so $IR(B_1),IR(B_2) \in \B$. Thus, $NewBounds$ only contains elements in $\B$, and the invariant is true at the end of each loop. Upon exiting the loop, $Bounds = \emptyset$, so the statement is trivial.

    (2) This statement initially holds since $IR(B_I) \subseteq B_I$ by Lemma 5.4. We now look at $NewBounds$ to see that the statement holds for it. At line 19, $Keep$ is either the union of two adjacent bounds $IR(B_1)$ and $IR(B_2)$, or it contains one or two sets $IR(B_i)$. In the former case, the invariant inductively tells us that $B \subseteq B_I$, so by correctness of \textsc{Decompose}, $B_1,B_2 \subseteq B$, and by $\phi(B_i)$, $IR(B_i) \subseteq B_i$. Thus, $IR(B_1) \cup IR(B_2) \subseteq B_1 \cup B_2 \subseteq B \subseteq B_I$. In the latter case, $\phi(B_i)$ gives that $IR(B_i) \subseteq B_i \subseteq B \subseteq B_I$. Line 22 is the only time that $Bounds$ is updated in this loop, and it is set to $NewBounds$, so the invariant is maintained at the end of each iteration for $Bounds$. Upon exiting the loop, the statement is trivially true since $Bounds = \emptyset$.

    (3) The statement is vacuously true initially ($Bounds$ does not contain any pair of distinct bounds). At line 19, it is guaranteed that the intervals in $Keep$ are non-adjacent by the if-statement's guard on line 13. If they were adjacent prior that if-statement, they would be joined into a new interval $IR(B_1) \cup IR(B_2) \subseteq B$, and since $B$ is non-adjacent with all other bounds in $Bounds$, the same holds for $IR(B_1) \cup IR(B_2)$. It is then evident that all intervals in $NewBounds$ are non-adjacent. Thus, since $Bounds$ is set to $NewBounds$ at line 22, the invariant is maintained at the end of each iteration of the outer loop. Upon exiting the outer loop, the statement is vacuously true.
\end{proof}

\paragraph{Proof of Theorem 6.2}
\begin{proof}
    First note that since $L$ satisfies ACC and DCC, $\ell(B)$ is defined for every $B \in \B$ and $\{F^k(B)\}$ is guaranteed to stabilize at some $k'$ (and thus any instantiation of $IR$ terminates). We prove termination by showing that $$\pi(Bounds) = \max_{B \in Bounds} \ell(B)$$ strictly decreases at each iteration, and the while loop terminates after $\pi(Bounds) = 0$.

    Consider a loop of the for loop beginning on line 6. First observe that at line 12, $\ell(B) > 0$ $\ell(B_1),\ell(B_2) < \ell(B)$ since if that were not the case then the guard of the if-statement on line 8 would be true (by interval properties, since any sub-interval with the same height or greater of the entire interval must be the entire interval), and the iteration would be skipped (by correctness of \textsc{Decompose}). After line 12, $$\pi(Keep) = \max\{\ell(IR(B_i)) \mid i \in \{1,2\} \wedge \phi(B_i)\} \leq \max\{\ell(B_i) \mid i \in \{1,2\} \wedge \phi(B_i)\} < \ell(B)$$ If $IR(B_1)$ and $IR(B_2)$ are adjacent (assuming they are both added) and $IR(B_1) \cup IR(B_2) = B$, then line 17 is taken and $B$ is not added back to $Bounds$. We can say that $B$ \emph{contributed} $-\ell(B)$ to the next iteration. If $IR(B_1) \cup IR(B_2) \neq B$, then it is still a strict subset, so $\ell(IR(B_1) \cup IR(B_2)) < \ell(B)$, and $B$ is effectively replaced by $IR(B_1) \cup IR(B_2)$. We say that $B$ contributed $\ell(IR(B_1) \cup IR(B_2)) - \ell(B) < 0$ to the next iteration in this case. If $IR(B_1)$ and $IR(B_2)$ are not adjacent, then $B$ is replaced by both of them. Since their heights are strictly smaller than $B$'s, we say that $B$ contributed $\pi(Keep) - \ell(B) < 0$ to the next iteration. If, w.l.o.g., only $IR(B_1)$ was kept, then $B$ contributed $\ell(B_1) - \ell(B) < 0$ to the next iteration. In all cases, $B$ contributed negative to the next iteration, and it is clear to see that since this is the case for each $B \in Bounds$, the max of $Bounds$ must strictly decrease. 

    If $\pi(Bounds) = 0$, then either $Bounds = \emptyset$, in which case the loop terminates, or $\forall B \in Bounds$, $\ell(B) = 0$, in which case they will all be removed from $Bounds$ as a result of the if-statement on line 8. On the next iteration, it will then necessarily be the case that $Bounds = \emptyset$ and the loop terminates.
\end{proof}

\paragraph{Proof of Theorem 6.3}
\begin{proof}
    Let $\mu \in \fix{f}$ and $\mu \in B_I$. Then $B_I \in \B_{sound}$ and by the proof of Lemma 5.3, $\mu \in IR(B_I)$. Upon entering the loop, $Bounds = \{IR(B_I)\}$, so the invariant is initially true.

    Now inductively let $\mu$ be a fixed point of $f$ contained in some bound $B$ of $Bounds$ (if there are none, then all fixed points are in $Final$ inductively by the invariant). Then, by correctness of \textsc{Decompose}, $\mu \in B_1$ or $\mu \in B_2$ after line 7. Assume w.l.o.g. that $\mu \in B_1$. If line 9 is reached, then it must be the case that $B_1 = B$, so after line 9 $\mu$ will be in a bound in $Final$ since $\mu \in B$. If line 9 is not reached, then after line 12 we have that $IR(B_1) \in Keep$ since $B_1$ is sound (so $\phi(B_1)$). By the proof of Lemma 5.3, we have that $\mu \in IR(B_1)$. By the if-else statement on line 16, either $\bigcup Keep$ is added to $Final$, in which case $\mu$ is now in a bound in $Final$, or $Keep$ is added to $NewBounds$, in which case $\mu$ is now in a bound in $NewBounds$. At line 22, $Bounds$ is set to $NewBounds$, so $\mu$ would then be in a bound in $Bounds$ by the end of the iteration. Thus the invariant holds at the end of each iteration.

    Upon exiting the loop, the invariant is still upheld and $Bounds = \emptyset$ so all fixed points are contained in a bound in $Final$.
\end{proof}

\end{document}